%% file: paper.tex
\documentclass[sigplan,screen]{acmart}
\usepackage{booktabs}   
\usepackage{subcaption} 
\usepackage[percent]{overpic}
\usepackage{array}
\usepackage{algorithmic}
\usepackage{graphicx}
\usepackage{textcomp}
\usepackage{xcolor}
\usepackage{xspace}
\usepackage{amsthm}
\usepackage{empheq}
\usepackage{mathtools}
\usepackage{pifont}
\usepackage{array}
\usepackage{multirow}
\usepackage{wrapfig}
\usepackage{listings}
\usepackage{bm}
\usepackage{enumitem}

\newtheorem*{theorem*}{Theorem}
\newtheorem{theorem}{Theorem}

\newtheorem*{corollary*}{Corollary}
\newtheorem{corollary}{Corollary}
\newtheorem*{lemma*}{Lemma}
\newtheorem{lemma}{Lemma}
\newtheorem*{definition*}{Definition}
\newtheorem{definition}{Definition}

\newcommand{\etc}{\emph{etc.}\xspace}
\newcommand{\eg}{\emph{e.g.}\xspace}
\newcommand{\ie}{\emph{i.e.}\xspace}
\newcommand{\etal}{\emph{et al.}\xspace}
\newcommand{\ro}{\textnormal{RO}\xspace}
\newcommand{\rz}{\textnormal{RZ}\xspace}

\newcommand{\ru}{\textnormal{RU}\xspace}
\newcommand{\rd}{\textnormal{RD}\xspace}
\newcommand{\rn}{\textnormal{RN}\xspace}
\newcommand{\rna}{\textnormal{RNA}}
\newcommand{\rne}{\textnormal{RNE}}
\newcommand{\faith}{\textnormal{FR}\xspace}
\newcommand{\rnd}{\circ\xspace}
\newcommand{\ufp}{\textnormal{ufp}}
\newcommand{\ulp}{\textnormal{ulp}}
\newcommand{\fpred}{\textnormal{pred}}
\newcommand{\fsucc}{\textnormal{succ}}

\setcopyright{none}
\renewcommand\footnotetextcopyrightpermission[1]{}

\begin{document}

\title{Odd but Error-Free FastTwoSum: More General Conditions for FastTwoSum as an Error-Free Transformation for Faithful Rounding Modes}
\subtitle{Rutgers Department of Computer Science Technical Report DCS-TR-761}

\author{Sehyeok Park}
\affiliation{%
  \institution{Rutgers University}
  \city{New Brunswick}
  \state{New Jersey}
  \country{USA}}
\email{sp2044@cs.rutgers.edu}

\author{Jay P. Lim}
\affiliation{%
  \institution{University of California, Riverside}
  \city{Riverside}
  \state{California}
  \country{USA}}
\email{jlim@ucr.edu}

\author{Santosh Nagarakatte}
\affiliation{%
  \institution{Rutgers University}
  \city{New Brunswick}
  \state{New Jersey}
  \country{USA}}
\email{santosh.nagarakatte@cs.rutgers.edu}

\input{sec.abstract}


\maketitle

\input{sec.intro}
\input{sec.background}
\input{sec.faithful}
\input{sec.rto}

\input{sec.applications}
\input{sec.related}
\input{sec.conclusion}

\newpage

\begin{acks}
This material is based upon work supported in part by the research
gifts from the Intel corporation and the
\grantsponsor{GS100000001}{National Science
  Foundation}{http://dx.doi.org/10.13039/100000001} with grants:
\grantnum{GS100000001}{2110861} and \grantnum{GS100000001}{2312220}.
  Any opinions, findings, and conclusions or recommendations expressed
  in this material are those of the authors and do not necessarily
  reflect the views of the Intel corporation or the National Science
  Foundation.
\end{acks}

\renewcommand\refname{References}
\bibliography{reference}

\end{document}

%% file: sec.abstract.tex
\begin{abstract}
  This paper proposes sufficient, yet more general conditions for
  applying FastTwoSum as an error-free transformation (EFT) under all
  faithful rounding modes. Additionally, it also identifies guarantees
  tailored to round-to-odd for establishing FastTwoSum as an EFT.
  This paper also describes a floating-point splitting tailored for
  round-to-odd that is an EFT where the distribution of bits is
  configurable (\ie, ExtractScalar for round-to-odd).  Our sufficient
  conditions are more general than those previously known in the
  literature (\ie, it applies to a wider operand domain).
\end{abstract}

%% file: sec.intro.tex
\section{Introduction}

Under select circumstances, the rounding error of a finite-precision,
floating-point addition is itself a floating-point number. FastTwoSum
is the most prominent method for computing floating-point rounding
errors using floating-point operations. Given a floating-point set
$\mathbb{F}$ and a rounding operation $\circ$ mapping $\mathbb{R}$ to
$\mathbb{F}$, FastTwoSum takes inputs $a, b \in \mathbb{F}$ and
produces outputs $x, y \in \mathbb{F}$ as shown below.

\begin{empheq}[box=\fbox]{align*}
\textbf{Fa}&\textbf{stTwoSum}(a, b):\\
  &x = \circ(a+b)\\
  &z = \circ(x-a)\\
  &y = \circ(b-z)\\
  &\textbf{return } x, y
\end{empheq}

FastTwoSum is the foundation for various algorithms designed to
improve the accuracy of finite-precision addition. Through the
operations shown above, FastTwoSum produces $x$, the floating-point
sum of its inputs using the rounding rule imposed by $\circ$. The
second output $y$ is an estimate of $\delta = a+b - x$: the rounding
error incurred in computing $x$. Algorithms based on FastTwoSum
utilize this error estimate to either compensate the initial
floating-point sum or extend the
precision past what is available in
$\mathbb{F}$~\cite{dekker:fts:1971}~\cite{priest:arbitrary-prec:arith:1991}
~\cite{shewchuk:adaptive-prec:1997}~\cite{hida:qd:2001}. 
Compensation and extended-precision algorithms can 
maximize the accuracy of their results by applying FastTwoSum as an 
\textit{error-free transformation (EFT)}, through which the original inputs 
are transformed to an equivalent pair of outputs (\ie, $x+y = a+b$).

Dekker's analysis of FastTwoSum~\cite{dekker:fts:1971} demonstrates
that $y$ is exactly equal to the rounding error $\delta$ given the
following conditions: $\mathbb{F}$'s base is $\beta \le 3$, $|a| \ge
|b|$, and $\circ$ applies round-to-nearest ($\rn$). The equivalence
between $y$ and $\delta = a+b - x$ implies $x+y = a+b$, which ensures
FastTwoSum is an EFT.

Later analyses of FastTwoSum have explored EFT guarantees beyond the
setting proposed by Dekker. In particular, several works have analyzed
FastTwoSum's exactness and rounding properties in the absence of
$\rn$~\cite{linnainmaa:fp-analysis:1974}~\cite{demmel:reprod-sum:arith:2013}
~\cite{graillat:num-val:2015}~\cite{boldo:fts-robustness:toms:2018}~\cite{jeannerod:fts-revisited:arith:2025}. The
rounding error induced by floating-point addition under $\rn$ is
guaranteed to be a floating-point number barring any
overflow~\cite{bohlender:exact-fp-ops:arith:1991}~\cite{daumas:exact-fp-computing:tphol:2001}~\cite{boldo:correcting-terms:arith:2003}. However,
the guarantee does not extend to other faithful rounding modes without
the \rn properties~\cite{boldo:correcting-terms:arith:2003}~\cite{muller:fp-handbook:2018}.

The IEEE-754 standard~\cite{ieee754:2019} supports multiple rounding
modes: round-to-nearest (ties-to-even ($\rne$) or ties-to-away
($\rna$)), round-down ($\rd$), round-up ($\ru$), and round-toward-zero
($\rz$). The latter three, which are collectively referred to as
directed rounding modes, do not guarantee that numbers are rounded to
their closest values in $\mathbb{F}$. While IEEE-754 enforces $\rne$
as the default rounding mode, the directed rounding modes have
multiple use cases including interval arithmetic. Due to the high cost
of switching rounding modes on modern machines and the emergence of
custom floating-point units without $\rne$ support, generalizing EFT
guarantees for FastTwoSum to other rounding modes is of great
interest.

Round-to-odd ($\ro$) is a faithful rounding mode that will be
supported by the upcoming P3109 standard for machine learning
arithmetic formats. Akin to $\rn$, $\ro$ is an unbiased rounding mode
for which the probabilistic mean value of the error for a single
instance of rounding is zero. Unlike $\rn$, $\ro$ makes double
rounding innocuous. Rounding to some intermediate, higher precision
format via $\ro$ prior to the final rounding produces the same result
as directly rounding to the destination
format~\cite{boldo:ro:2005}~\cite{russinoff:verif-fp-hardware:2019}. 
An interesting application of $\ro$ is the development of multi-precision math
libraries~\cite{lim:rlibm-all:popl:2022}~\cite{park:rlibm-multi:2025}. 
Given the anticipated increase in $\ro$ adoption due to upcoming standards, 
establishing EFT guarantees for FastTwoSum under $\ro$ will be of practical use 
in the design of future algorithms.

FastTwoSum is an EFT when the operation $z = \circ(x - a)$ is exact
and the rounding error $\delta = a + b - x$ is representable in
floating-point. This paper establishes that both properties are
guaranteed under all faithful rounding modes if (i) $a$ is
an integer multiple of $\ulp(b)$ and (ii) $b$ is an
integer multiple of $2u^{2} \cdot \ufp(a)$, where $u$, $\ufp(\cdot)$, 
and $\ulp(\cdot)$ denote the unit round-off, unit in the first place, and unit 
in the last place, respectively. Given $p$-bits of available precision, these 
conditions can potentially tolerate up to a $2p-1$ exponent difference between 
a and b. For operands that satisfy our conditions, the permitted exponent 
difference nearly doubles that of previously known EFT conditions for 
faithfully rounded FastTwoSum, making our conditions applicable to a wider 
input domain. In addition to these conditions, which suffice for all faithful 
rounding modes including $\ro$, we present new EFT guarantees specifically 
tailored to $\ro$. We apply these conditions to the
ExtractScalar algorithm - a FastTwoSum-based EFT for $\rn$ that splits
a single floating-point input across two numbers while maintaining the
original value~\cite{rump:accsum-1:siam:2008}. Specifically, we design a new 
variant of ExtractScalar
that preserves its original properties under $\ro$.

%% file: sec.background.tex
\section{Background}\label{sec:background}

\subsection{Definitions}

We denote by $\mathbb{F}$ the set of floating-point numbers of base
$\beta = 2$, precision $p \ge 2$, and extremal exponents $E_{min}$ and
$E_{max}$. We additionally include infinities in $\mathbb{F}$. 
For all reals $x \in \mathbb{F}$, $x = M_{x} \times 2^{E_{x}-p+1}$ 
for a unique pair of $M_{x}, E_{x} \in \mathbb{Z}$ that meet the following 
conditions.

\begin{center}
\begin{minipage}{.3\linewidth}
\begin{align*}
  E_{min} \le E_{x} \le E_{max}
\end{align*}
\end{minipage}
\begin{minipage}{.3\linewidth}
\begin{align*}
  |M_{x}| < 2^{p}
\end{align*}
\end{minipage}
\begin{minipage}{.3\linewidth}
\begin{align*}
  E_{x} > E_{min} \Rightarrow |M_x| \ge 2^{p-1}
\end{align*}
\end{minipage}
\end{center}

A nonzero real $x \in \mathbb{F}$ is normal if $|x| \ge 2^{E_{min}}$ 
and subnormal otherwise. We denote the largest finite number
in $\mathbb{F}$ by $\Omega = (2^{p} - 1) \times 2^{E_{max} - p +
  1}$. Similarly, $\omega = 2^{E_{min}-p+1}$ represents the smallest
nonzero magnitude.

We express through $\circ \in \rn$ that $\circ$ applies $\rn$ with any
arbitrary tie breaking rule (\eg, ties-to-nearest, ties-to-away). We
apply the notation $\circ \in \faith$
from~\cite{jeannerod:fts-revisited:arith:2025} to indicate that
$\circ$ performs faithful rounding (\ie, $\forall r \in \mathbb{R},
\circ(r) \in \{\rd(r), \ru(r)\}$). We use $\circ(a \pm b)$ to denote
floating-point addition or subtraction between $a, b \in \mathbb{F}$
under the rounding mode specified by $\circ$. We denote by $u =
2^{-p}$ the unit round-off, which is the distance between 1 and its
closest number in $\mathbb{F}$. For a given $r \in \mathbb{R}$, we
define its exponent $e_{r}$, unit in the first place ($\ufp(r)$), and
unit in the last place ($\ulp(r)$) as shown below.

\begin{definition}
  \label{def:ufp-ulp}
For $r \in \mathbb{R}$, 

\begin{center}
\begin{minipage}{.3\linewidth}
\begin{align*}
	e_r = \begin{cases}
		-\infty, &\text{if} \ r = 0\\
		\lfloor \log_{2}(|r|) \rfloor &\text{otherwise}
	\end{cases}
\end{align*}
\end{minipage}
\begin{minipage}{.4\linewidth}
\begin{align*}
	\ufp(r) = 2^{e_r}	
\end{align*}
\end{minipage}
\begin{minipage}{.9\linewidth}
\begin{align*}
	\ulp(r) = 
	\begin{cases}
		2u \cdot \ufp(r), &\text{if} \ |r| \ge 2^{E_{min}}\\
		\omega &\text{otherwise}
	\end{cases}
\end{align*}
\end{minipage}
\end{center}
\end{definition}

Definition~\ref{def:ufp-ulp} ensures for all finite $x \in \mathbb{F}$
that $\ulp(x) \ge 2u \cdot \ufp(x)$. Moreover, for all finite $x, y
\in \mathbb{F}$, $\ufp(x) \ge \ufp(y)$ implies $\ulp(x) \ge \ulp(y)$.

We denote by $\fpred(r)$ the \textit{floating-point predecessor} of $r
\in \mathbb{R}$. Conversely, $\fsucc(r)$ represents the
\textit{floating-point successor}.

\begin{definition}
  \label{def:pred-succ-real}
For $r \in \mathbb{R}$ and $r\notin \mathbb{F}$,
	\begin{align*}
		&\fpred(r) = \max\{x \in \mathbb{F} \ \vert \ x < r \}\\
		&\fsucc(r) = \min\{x \in \mathbb{F} \ \vert \ x > r \}
	\end{align*}
\end{definition}\label{def:pred-succ}

We collectively refer to $\fpred(r)$ and $\fsucc(r)$ as $r$'s floating-point neighbors. For finite numbers in $\mathbb{F}$, we define their floating-point neighbors as follows.

\begin{definition}\label{def:pred-succ-fp}
For finite $x \in \mathbb{F}$,
\begin{align*}
&\fpred(x) = 
\begin{cases}
  -\omega, &\text{if} \ x = 0\\
	-\infty, &\text{if} \ x = -\Omega\\
	x - \frac{1}{2}\ulp(x), &\text{if} \ x = \ufp(x) \ \text{and} \ 
	x > 2^{E_{min}}\\
	x - \ulp(x) &\text{otherwise}
\end{cases}\\
&\fsucc(x) = 
\begin{cases}
  \omega, &\text{if} \ x = 0\\
	\infty, &\text{if} \ x = \Omega\\
	x + \frac{1}{2}\ulp(x), &\text{if} \ x = -\ufp(x) \ \text{and} \ x < -2^{E_{min}}\\
	x + \ulp(x) &\text{otherwise}
\end{cases}
\end{align*}
\end{definition}

\subsection{Basic Properties of Floating-Point}

We present the basic floating-point properties most relevant to our theorems.
Henceforth, the notation $x \in y\mathbb{Z}$ for $x, y \in \mathbb{R}$ expresses that $x$ is an integer multiple of $y$.  Every finite, nonzero $x \in \mathbb{F}$ satisfies $\ufp(x) \le |x| < 2\ufp(x)$. Every finite $x \in \mathbb{F}$ also satisfies $x \in \omega\mathbb{Z}$ and $x \in \ulp(x)\mathbb{Z}$. If $x \in 2^{k}\mathbb{Z}$ for $k \in \mathbb{Z}$, then $x \in 2^{i}\mathbb{Z}$ for any $i \in \mathbb{Z}$ less than $k$. Because $2u \cdot \ufp(x)$ is an integer power of 2 and Definition~\ref{def:ufp-ulp} guarantees $\ulp(x) \ge 2u \cdot \ufp(x)$, it follows that $x \in 2u \cdot \ufp(x)\mathbb{Z}$. By extension, $x$ is an integer multiple of any smaller integer power of 2 (\eg, $u \cdot \ufp(x)$, $2u^{2} \cdot \ufp(x)$, \etc).

We assume throughout the remainder of the paper that the operands of FastTwoSum are finite (\ie, $|a|, |b| \le \Omega$). All finite $a, b \in \mathbb{F}$ satisfy $a, b \in \min(\ulp(a), \ulp(b))\mathbb{Z}$. Because addition and subtraction preserve common factors, the real sum $a+b$ and its floating-point counterpart $\circ(a+b)$ satisfy $a+b, \circ(a+b) \in \min(\ulp(a), \ulp(b))\mathbb{Z}$ for all $\circ \in \faith$. Consequently, the resulting rounding error $\delta = a+b - \circ(a+b)$ adheres to $\delta \in \min(\ulp(a), \ulp(b))\mathbb{Z}$. Because $a, b \in \omega \mathbb{Z}$, it also follows that $a+b, \circ(a+b), \delta \in \omega \mathbb{Z}$. This property implies neither the sum of numbers in $\mathbb{F}$ nor the associated rounding error is subject to underflow, meaning these values are integer multiples of $\omega$.

A number $r \in \mathbb{R}$ is in $\mathbb{F}$ if there exist $M_{r}, E_{r} \in 
\mathbb{Z}$ such that $r = M_{r} \times 2^{E_r - p + 1}$, $|M_{r}| < 2^{p}$, 
and $E_{min} \le E_{r} \le E_{max}$. We provide the conditions sufficient to 
meet these constraints in the following lemma.

\begin{lemma}\label{lemma:fp}
Let $r \in \mathbb{R}$, $k \in \mathbb{Z}$, and $\sigma = 2^{k}$. If $|r| \le 
\min(\sigma, \Omega)$ and $r \in \max(u \cdot \sigma, \omega)\mathbb{Z}$, then 
$r \in \mathbb{F}$.
\end{lemma}

If $a, b \in \mathbb{F} \setminus \{\pm \infty\}$, then $\delta = a + b - 
\circ(a+b) \in \omega\mathbb{Z}$ for all $\circ \in \faith$. Therefore, $\delta 
\in u \cdot 2^{k} \mathbb{Z}$ for $k \in \mathbb{Z}$ implies $\delta \in \max(u 
\cdot 2^{k}, \omega)\mathbb{Z}$.

\begin{corollary}\label{corollary:fp}
Let $a, b \in \mathbb{F} \setminus \{ \pm \infty \}$ and $k \in \mathbb{Z}$. 
Let $\delta = a+b - \circ(a+b)$ and $\sigma = 2^{k}$. If $|\delta| \le \min(\sigma, 
\Omega)$ and $\delta \in u \cdot \sigma\mathbb{Z}$, then $\delta \in \mathbb{F}$.
\end{corollary}

Lastly, we summarize the key properties of rounding. For all $r \in \mathbb{R}$ such that $|r| \le \Omega$, faithful rounding guarantees $|\rnd(r) - r| \le \frac{1}{2}\ulp(r)$ when $\circ \in \rn$ and $|\circ\!(r) - r| < \ulp(r)$ otherwise. Faithful rounding also enforces sign preservation: $\circ(r)\cdot r  \ge 0$ for all $\circ \in \faith$ and $r \in \mathbb{R}$. 

\subsection{Properties of Round-to-Odd}

For base $\beta = 2$, round-to-odd (\ro) is a faithful rounding mode that maps 
all numbers $r \not\in \mathbb{F}$ to a floating-point neighbor $x \in \{\fpred(r), \fsucc(r)\}$ such that $x$'s 
binary encoding forms an \textit{odd integer}. For all numbers $r \in 
\mathbb{R}$ and $\mathbb{F}$ such that $p \ge 2$, we define $\ro(r)$ as 
follows. 

\begin{definition}\label{def:rto}
For $r \in \mathbb{R}$ and $\mathbb{F}$ such that $p \ge 2$, 

\begin{align*}
\ro(r) = 
\begin{cases}
	r, &\text{if} \ r \in \mathbb{F}\\
  \fsucc(r), &\text{if} \ r \not \in \mathbb{F} \ \text{and} \ \fsucc(r) 
  \ \text{has odd encoding}\\
  \fpred(r) &\text{if} \ r \not \in \mathbb{F} \ \text{and} \ \fpred(r) 
  \ \text{has odd encoding}
\end{cases}
\end{align*}
\end{definition}

Due to Definition~\ref{def:rto}, $x = \ro(r) = M_{x} \times 2^{E_{x}-p+1}$ has an \textit{even significand} (\ie, $M_{x}$ is an even integer) only if $r \in \mathbb{F}$. If $p \ge 2$, all finite $x \in \mathbb{F}$ such that $|x| = \ufp(x) > \omega$ (\ie, $|x|$ is an integer power of 2 greater than $\omega$) have even significands. Therefore, \ro does not renormalize within the dynamic range: for all $r \in \mathbb{R}$ such that $\omega \le |r| \le \Omega$, $x = \ro(r) \in \mathbb{F}$ satisfies $e_{x} = e_{r}$. Definition~\ref{def:rto} also implies $\ro \not\in \rn$; hence, $|\ro(r) - r| < \ulp(r)$ for all reals $r$ such that $|r| \le \Omega$.

\subsection{Exactness Conditions for FastTwoSum}

Given finite inputs $a, b \in \mathbb{F}$ such that $|a+b|\le \Omega$ and the 
operations $x = \circ_{1}(a+b)$, $z = \circ_{2}(x - a)$, and 
$y = \circ_{3}(b - z)$ where $\circ_{1}, \circ_{2}, \circ_{3} \in \faith$, we 
present below the properties that ensure $x + y = a + b$ (\ie, FastTwoSum is an 
EFT). While previous works on the topic denote the rounding error induced by 
floating-point addition (\ie, $a+b - x$) as   
$e$~\cite{jeannerod:fts-revisited:arith:2025}, we refer to said 
error throughout the rest of the paper using $\delta$.

\begin{center}
\begin{itemize}[leftmargin=*]
	\item[]{\textbf{Property 1.}} $x - a \in \mathbb{F}$
  \item[]{\textbf{Property 2.}} $\delta = a+b - x \in \mathbb{F}$
\end{itemize}
\end{center}

If Property 1 holds, then $z = \circ_{2}(x - a) = x - a$ for all $\circ_{2} \in \faith$. The equality $z = x - a$ implies $y = \circ_{3}(b - z) = \circ_{3}(a + b - x)$, the latter of which is the \textit{correct rounding} of the rounding error $\delta = a+b - x$ under $\circ_{3}$. Properties 1 and 2 thus jointly induce $y = a+b - x$, thereby ensuring the desired equality $x+y = a+b$. We present below an overview of previously established sufficient conditions for each property.

\textbf{Conditions for Property 1.} 
Jeannerod and Zimmermann prove that $a \in \ulp(b)
\mathbb{Z}$ is sufficient for all faithfully rounded sums 
$\circ_{1}(a+b)$~\cite[Lemma 2]{jeannerod:fts-revisited:arith:2025}. We recall 
henceforth that  $a \in \ulp(b)\mathbb{Z}$ ensures $z = x - a$ for all 
$\circ_{1}, \circ_{2} \in \faith$.

\textbf{Conditions for Property 2.} If $\circ \in \rn$, then $\delta = a+b - 
\circ(a+b)$ must be an element of $\mathbb{F}$. This guarantee, however, could 
fail if $a+b$ is not rounded to its closest neighbor (\ie, $|\delta| > 
\frac{1}{2}\ulp(a+b)$). In particular, $\delta$ may not be in $\mathbb{F}$ when 
the exponent ranges occupied by the operands \textit{do not overlap} 
(\eg, $\ulp(a) > |b|$). Suppose $|b| < \frac{1}{2}\ulp(a)$. 
Under \rn, $\circ(a+b) = a$ and $\delta = b$. Because $b 
\in \mathbb{F}$, the rounding error is by default an element of $\mathbb{F}$. 
If $\circ(a+b)$ is not the nearest neighbor of $a+b$ in $\mathbb{F}$, however, 
$\delta$ could be $b \pm \frac{1}{2}\ulp(a)$ or $b \pm \ulp(a)$ (see 
Definition~\ref{def:pred-succ-fp}). In such cases, the final rounding error may
not be exactly representable in $\mathbb{F}$. For example, suppose $a = 2^{p}$ 
and $b = 2^{-p}$. In this example, $|b| < \frac{1}{2}\ulp(a) = 2^{0}$. If 
$\circ = \ru$, $\circ(a+b) = 2^{p} + 2 = a + \ulp(a)$. The resulting rounding 
error is $a+b - \circ(a+b) = 2^{-p} - 2$, and thus $\delta \not\in \mathbb{F}$.

For nonzero $a, b \in \mathbb{F}$, Boldo \etal present $|e_{a} - e_{b}| \le p - 
1$ as a sufficient condition for Property 2 under all $\circ \in 
\faith$~\cite[Lemma 2.6]{boldo:fts-robustness:toms:2018}. Jeannerod and 
Zimmermann prove in ~\cite[Lemma 1]{jeannerod:fts-revisited:arith:2025} that 
the less restrictive conditions $a \in \ulp(b)\mathbb{Z}$ and $e_{a} - e_{b} 
\le p$ suffice. We recall henceforth that $a \in \ulp(b)\mathbb{Z}$ and 
$e_{a} - e_{b} \le p$ ensure $\delta \in \mathbb{F}$ for all $\circ \in \faith$. 
In Section~\ref{sec:faithful}, we establish new conditions that 
correctly ensure $\delta \in \mathbb{F}$ for all faithful rounding modes, 
\textit{even when $|e_{a} - e_{b}|$ exceeds $p$}.

For directed rounding modes, the anticipated rounding direction (\ie the expected sign of $\delta$) can influence $\delta$'s representability. As before, consider operands $a = 2^{p}$ and $b = 2^{-p}$. Under $\rd$, $\delta$ \textit{must be non-negative} (\ie $a+b \ge \rd(a+b)$). Subsequently, $b \ge 0$ would imply $\rd(a+b) = a$ and $\delta = a+b - a = b \in \mathbb{F}$. Likewise, $b \le 0$ is sufficient to guarantee $\delta \in \mathbb{F}$ under $\ru$. While assuming $a \in \ulp(b)\mathbb{Z}$, Jeannerod and Zimmermann prove $b \ge 0$, $b \le 0$, and $a \times b \ge 0$ are sufficient for $\rd$, $\ru$, and $\rz$ respectively~\cite{jeannerod:fts-revisited:arith:2025}. The condition $a \times b \ge 0$ for $\rz$ also appears in~\cite{linnainmaa:fp-analysis:1974}, which assumes $|a| \ge |b|$. In Section~\ref{sec:rto}, we introduce conditions tailored to $\ro$ that \textit{do not restrict the signs of the operands}.

%% file: sec.faithful.tex
\section{EFT Conditions for Faithful Rounding}\label{sec:faithful}

To establish EFT guarantees for FastTwoSum, we first identify
conditions that ensure $\delta = a+b - \circ(a+b) \in \mathbb{F}$ for
all $\circ \in \faith$. For this purpose, we present a method of
determining $\delta$'s membership based on the sum's magnitude.

\begin{lemma}\label{lemma:faith-1}
Let $p \ge 2$ and $x = \circ(a+b)$. If both conditions
\begin{itemize}
	\item[]{(i)} $\circ \in \faith$
	\item[]{(ii)} $|a+b| \le 2^{E_{min}+1}$
\end{itemize}
are satisfied, then $\delta = a+b - x = 0$.
\end{lemma}

\begin{proof}
	All finite numbers in $\mathbb{F}$ are integer multiples of $\omega$.
	Addition and subtraction preserve common factors, so $a+b \in 
	\omega\mathbb{Z}$. Given $\omega = u \cdot 2^{E_{min}+1}$,
	the bound $|a+b| \le 2^{E_{min+1}}$ and $a+b \in \omega\mathbb{Z}$ 
	collectively imply $a+b \in \mathbb{F}$ due to Lemma~\ref{lemma:fp}.
	Since $x$ is a faithful rounding of $a+b$, it follows that $x = a+b$ and 
	$\delta = 0$.
\end{proof}

Lemma~\ref{lemma:faith-1} signifies that the addends' magnitudes can 
induce implicit bounds on their sum that guarantee $\delta = 0$, which then 
ensures $\delta \in \mathbb{F}$. We proceed with exploring sufficient 
conditions for $\delta \in \mathbb{F}$ by leveraging the \textit{relative} 
magnitudes of $a$ and $b$.

\begin{lemma}\label{lemma:faith-2}
Let $p \ge 2$ and $x = \circ(a+b)$. If all conditions
\begin{itemize}
	\item[]{(i)} $\circ \in \faith$
	\item[]{(ii)} $|a+b| \le \Omega$
	\item[]{(iii)} $|a| \ge u \cdot \ufp(b)$
	\item[]{(iv)} $|b| \ge u \cdot \ufp(a)$
\end{itemize}
are satisfied, then $\delta = a+b - x \in \mathbb{F}$.
\end{lemma}

\begin{proof}
	Condition (ii) ensures that overflow does not occur for any
	$\circ \in \faith$ (\ie, $|x| \le \Omega$). Lemma~\ref{lemma:faith-1} 
	addresses all cases in which $|a+b| \le 2^{E_{min}+1}$. We thus assume for 
	the remainder of the proof that $|a+b| > 2^{E_{min}+1}$. Because
	$\delta \in \mathbb{F}$ is trivially true when $a = 0$ or $b = 0$, 
	we also disregard such cases.
	
	Suppose $|a| \ge |b|$. By the definition of $\ufp$, $|a| \ge
	|b|$ implies $|a| \ge \ufp(b)$. The bound $|a| \ge \ufp(b)$
	subsequently induces $|a| \ge u \cdot \ufp(b)$, thereby
	satisfying Condition (iii). The alternative assumption $|b|
	> |a|$ would have similarly induced $|b| > u \cdot
	\ufp(a)$, which signifies \textit{either Condition (iii) or
	(iv) trivially holds} for all operands.  

	For all finite $a \in \mathbb{F}$, $a \in \ulp(a)\mathbb{Z}$. 
	If $|a| \ge |b|$, then $\ulp(a) \ge \ulp(b)$ and 
	$a \in \ulp(b)\mathbb{Z}$. Condition (iv) induces $\ufp(b) \ge 
	u \cdot \ufp(a)$, which subsequently implies $e_{a} - e_{b} \le p$. 
	It then follows from~\cite[Lemma 1]{jeannerod:fts-revisited:arith:2025} 
	that $\delta \in \mathbb{F}$. Alternatively, $|b| > |a|$ and Condition (iii) 
	would imply $b \in \ulp(a)\mathbb{Z}$ and $e_{b} - e_{a} \le p$ 
	respectively. We can thus swap $a$ and $b$ and 
	apply~\cite[Lemma 1]{jeannerod:fts-revisited:arith:2025} to complete the 
	proof.
\end{proof}

Conditions (iii) and (iv) each ensure $\ufp(a) \ge u \cdot \ufp(b)$
and $\ufp(b) \ge u \cdot \ufp(a)$. If $|a| > 0$, Condition (iii) is
equivalent to $e_{b} - e_{a} \le p$. Similarly, Condition (iv) is
equivalent to $e_{a} - e_{b} \le p$ when $|b| > 0$. Assuming both
operands are nonzero, the conjunction of Conditions (iii) and (iv) is
subsequently equivalent to $|e_{a} - e_{b}| \le p$. The bound $|e_{a}
- e_{b}| \le p$ is comparable to the conditions $a \in
\ulp(b)\mathbb{Z}$ and $e_{a} - e_{b} \le p$ proposed
in~\cite{jeannerod:fts-revisited:arith:2025}. For any finite $x \in
\mathbb{F}$ such that $|x| \ge 2^{E_{min}}$, $\ulp(x) = 2u \cdot
\ufp(x)$ due to Definition~\ref{def:ufp-ulp}. As such, $|e_{a} -
e_{b}|$ could be \textit{$p$ or greater} when $\ulp(\max(|a|,
|b|)) > \min(|a|, |b|)$. Effectively, Lemma~\ref{lemma:faith-2}
requires that the exponent difference is at most $p$ \textit{when the
  operands' significands do not overlap}.

Although the conditions in Lemma~\ref{lemma:faith-2} are
\textit{sufficient} to ensure $\delta \in \mathbb{F}$, adherence to
the bound $|e_{a} - e_{b}| \le p$ is \textit{not necessary}. For
example, consider operands $a = 2^{p}$ and $b = 2^{-1}$ for which
$|a+b| \le \Omega$ and $e_{a} - e_{b} = p + 1$. If $\circ = \ru$, the
floating-point sum $x = \circ(a+b)$ produces $2^{p} + 2$. The
associated rounding error $\delta = (2^{p}+2^{-1}) - (2^{p}+2) =
2^{-1} - 2$ thus satisfies $\delta \in \mathbb{F}$ for precision $p
\ge 2$. For $\circ = \rd$, $x = 2^{p} = a$. The resulting rounding
error is $\delta = b$, which is by default in $\mathbb{F}$. If $\circ
\in \faith$, then $x \in \{\rd(a+b), \ru(a+b)\}$, and thus $\delta \in
\mathbb{F}$ under all faithful rounding modes in this example.

The previous example exhibits that $\delta$ can be an element of
$\mathbb{F}$ for all $\circ \in \faith$ even when $|e_{a} - e_{b}| >
p$. We present an additional set of less restrictive conditions that
guarantee $\delta \in \mathbb{F}$ for operands that would otherwise be
excluded by Lemma~\ref{lemma:faith-2}.

\begin{theorem}\label{theorem:faith-1}
Let $p \ge 2$ and $x = \circ(a+b)$. If all conditions
\begin{itemize}
	\item[]{(i)} $\circ \in \faith$
	\item[]{(ii)} $|a+b| \le \Omega$
	\item[]{(iii)} $a \in 2u^{2} \cdot \ufp(b)\mathbb{Z}$
	\item[]{(iv)} $b \in 2u^{2} \cdot \ufp(a)\mathbb{Z}$
\end{itemize}
are satisfied, then $\delta = a+b - x \in \mathbb{F}$.
\end{theorem}

\begin{proof}
	Given Condition (ii), $x = \circ(a+b)$ is finite for all $\circ \in
        \faith$. As before, we assume $a \neq 0$ and $b \neq 0$.
        We also assume $|a+b| > 2^{E_{min}+1}$ as Lemma~\ref
        {lemma:faith-1} is sufficient otherwise. Because $a \in 2u^{2} \cdot
        \ufp(a)\mathbb{Z}$ for all $a \in \mathbb{F}$, $|a| \ge |b|$
        ensures $\ufp(a) \ge \ufp(b)$ and $a \in 2u^{2} \cdot
        \ufp(b)\mathbb{Z}$. Similarly, $|b| > |a|$ guarantees
        $\ufp(b) \ge \ufp(a)$ and $b \in 2u^{2} \cdot
        \ufp(a)\mathbb{Z}$. Hence, all finite, nonzero $a, b \in \mathbb{F}$
        trivially satisfy at least one of Conditions (iii) and
        (iv). Without loss of generality, we henceforth assume $|a|
        \ge |b|$. We proceed with the proof by decomposing Condition
        (iv) into two cases: $|b| \ge u \cdot \ufp(a)$ and $|b| < u
        \cdot \ufp(a)$.

	\textbf{Case 1: $\bm{b \in 2u^{2}\cdot \ufp(a)\mathbb{Z}}$ and
        $\bm{|b| \ge u \cdot \ufp(a)}$.} Note that Conditions (i)
        and (ii) are identical to their counterparts in
        Lemma~\ref{lemma:faith-2}. Because $|b| \ge u \cdot \ufp(b)$
        for all finite $b \in \mathbb{F}$, $|a| \ge |b|$ guarantees
        $|a| \ge u \cdot \ufp(b)$. The assumptions $|a| \ge |b|$ and
        $|b| \ge u \cdot \ufp(a)$ therefore jointly satisfy the
        conditions in Lemma~\ref{lemma:faith-2}, which subsequently
        ensures $\delta \in \mathbb{F}$.

	\textbf{Case 2: $\bm{b \in 2u^{2}\cdot \ufp(a)\mathbb{Z}}$ and
          $\bm{|b| < u \cdot \ufp(a)}$.} Since $|a| \ge |b|$ and $|b| < u
        \cdot \ufp(a)$, $|a+b| < \fsucc(|a|)$ due to
        Definition~\ref{def:pred-succ-fp}. From the bound on $|a+b|$,
        we derive $\ufp(a+b) \le \ufp(a)$. Given the expected error
        bound $|\delta| < \ulp(a+b) = 2u \cdot \ufp(a+b)$ under
        faithful rounding, $|\delta|$ is thus strictly less than $2u
        \cdot \ufp(a)$. Since $a$ immediately satisfies $a \in 2u^{2}
        \cdot \ufp(a) \mathbb{Z}$, we derive $\delta \in 2u^{2} \cdot
        \ufp(a) \mathbb{Z}$ through Condition (iv). Through
        Corollary~\ref{corollary:fp}, the bound $|\delta| < 2u \cdot
        \ufp(a)$ and $\delta \in 2u^{2} \cdot \ufp(a)\mathbb{Z}$
        collectively ensure $\delta \in \mathbb{F}$.
	
	If $|b| > |a|$, $a$ and $b$ immediately satisfy Condition
        (iv). We can prove $\delta \in \mathbb{F}$ for $|b| > |a|$
        by decomposing Condition (iii) into $|a| \ge u \cdot \ufp(b)$
        and $|a| < u \cdot \ufp(b)$ as has been shown for Condition
        (iv).

\end{proof}

Any operands that satisfy the conditions in Lemma~\ref{lemma:faith-2}
immediately satisfy those in Theorem~\ref{theorem:faith-1}. While we
omit the proof for brevity, one can easily deduce $|a| \ge u \cdot
\ufp(b)$ and $|b| \ge u \cdot \ufp(a)$ each imply $a \in 2u^{2} \cdot
\ufp(b)\mathbb{Z}$ and $b \in 2u^{2} \cdot \ufp(a)\mathbb{Z}$. It then
follows that Theorem~\ref{theorem:faith-1} guarantees $\delta \in
\mathbb{F}$ for all nonzero $a, b \in \mathbb{F}$ such that $|e_{a} -
e_{b}| \le p$.

Due to Definition~\ref{def:ufp-ulp}, the $\ulp$ of a given number $x$
is bounded by $\ulp(x) \ge 2u \cdot \ufp(x)$. For any $x \in
\mathbb{F}$ such that $2u^{2}\cdot \ufp(x) \ge \omega$, the term
$2u^{2} \cdot \ufp(x)$ as used in Conditions (iii) and (iv) is
effectively the $\ulp$ of $x$ \textit{if $\mathbb{F}$ had twice the
  available precision}. Consequently, the exponent difference for
operands that satisfy these conditions could potentially exceed
$p$. We highlight the implications of Theorem~\ref{theorem:faith-1}
through the following example.

Suppose $b$ is exactly equal to $2u^{2} \cdot \ufp(a)$, thereby
satisfying $b \in 2u^{2} \cdot \ufp(a)\mathbb{Z}$. If $b = 2u^{2}
\cdot \ufp(a)$, then $\ufp(b) = 2u^{2} \cdot \ufp(a)$. The latter
equality ensures $2 u^{2} \cdot \ufp(b) < 2 u^{2} \cdot
\ufp(a)$. Since $a$ immediately satisfies $a \in 2u^{2} \cdot
\ufp(a)\mathbb{Z}$, $b = 2u^{2} \cdot \ufp(a)$ indicates $a \in 2u^{2}
\cdot \ufp(b)\mathbb{Z}$. Assuming $|a+b| \le \Omega$, the operands in
this example meet all conditions in
Theorem~\ref{theorem:faith-1}. Since $u = 2^{-p}$, $\ufp(b) = 2u^{2}
\cdot \ufp(a)$ implies $\ufp(b) = 2^{1-2p}\cdot \ufp(a)$. The exponent
difference $e_{a} - e_{b}$ in this example is thus $2p -1$, which is
nearly double the bound enforced by
Lemma~\ref{lemma:faith-2}. Theorem~\ref{theorem:faith-1}'s conditions
are thus significantly less restrictive than those in
Lemma~\ref{lemma:faith-2} and are applicable to a wider domain.

Theorem 1's conditions also ensure that the real sum $a+b$ is exactly
representable in $2p$-bits of precision. Suppose $|a| > |b|$ and
$e_{a} - e_{b} = p + k$, where $0 \le k \le p-1$. In this example,
Condition (iv) would ensure that $b$ has at least $k$ trailing
zeroes. Consequently, $b$ will have at most $p - k$ effective bits
(\ie, available precision minus trailing zeroes). Hence, $a+b$ would
require at most $e_a - e_b + (p - k) = 2p$ bits to represent.

We now explore conditions that guarantee FastTwoSum is an EFT. As
described in Section~\ref{sec:background}, $x - a \in \mathbb{F}$
ensures $z = \circ_{2}(x - a) = x - a$ and $y = \circ_{3}(a + b - x) =
\circ_{3}(\delta)$. Therefore, $\delta \in \mathbb{F}$ is sufficient
for $x + y = a + b$ when $x - a \in \mathbb{F}$. We establish
properties that guarantee $x+y = a+b$ by combining the respective
requirements for $x - a \in \mathbb{F}$ and $\delta \in \mathbb{F}$.

\begin{theorem}\label{theorem:faith-2}
Let $p \ge 2$, $x = \circ_{1}(a+b)$, $z = \circ_{2}(x-a)$, and $y = \circ_{3}(b-z)$. If all conditions
\begin{itemize}
	\item[]{(i)} $\circ_{1}, \circ_{2}, \circ_{3} \in \faith$
	\item[]{(ii)} $|a+b| \le \Omega$
	\item[]{(iii)} $a \in \ulp(b)\mathbb{Z}$
	\item[]{(iv)} $b \in 2u^{2} \cdot \ufp(a)\mathbb{Z}$
\end{itemize}
are satisfied, then $x+y = a+b$. 
\end{theorem}

\begin{proof}
	As mentioned in Section~\ref{sec:background}, Condition (iii)
        guarantees $x - a \in \mathbb{F}$ and $z = \circ_{2}(x - a) =
        x-a$ for all $\circ_{1}, \circ_{2} \in \faith$. Because the
        definitions of $\ufp$ and $\ulp$ imply $\ulp(b) \ge 2u^{2}
        \cdot \ufp$(b), any operands that meet Condition (iii)
        implicitly satisfy $a \in 2u^{2} \cdot \ufp(b)\mathbb{Z}$
        (\ie, Condition (iii) in Theorem~\ref{theorem:faith-1}). All
        $a, b \in \mathbb{F}$ that meet Conditions (i) through (iv)
        therefore satisfy the conditions in
        Theorem~\ref{theorem:faith-1}, which confirms $\delta \in
        \mathbb{F}$ is valid for all $\circ_{1} \in \faith$. If $z =
        x-a$ and $\delta = a+b - x \in \mathbb{F}$, then $y =
        \circ_{3}(a+b - x) = a+b-x$ for all $\circ_{3} \in
        \faith$. Conditions (i) through (iv) thus suffice to conclude
        $x+y = a+b$.
\end{proof}

Conditions (iii) and (iv) each impose a lower bound on the magnitude
of $a$ and $b$ respectively when one is less than the other. If $|a|
\ge |b|$, $a$ and $b$ immediately satisfy Condition (iii). When $|a| <
|b|$, Condition (iii) requires any nonzero $|a|$ to be no less than
$\ulp(b)$. Since $\ulp(b) \ge 2u \cdot \ufp(b)$, any nonzero operands
that meet Condition (iii) are subject to the bound $e_{b} - e_{a} \le
p-1$. Similarly, Condition (iv) requires that $|b|$
satisfies $|b| \ge 2u^{2} \cdot \ufp(a)$ when $|a| > |b|$. The nonzero
operands that meet Condition (iv) thus satisfy $e_{a} - e_{b} \le 2p -
1$. Theorem~\ref{theorem:faith-2} thus maintains the \textit{maximum
possible exponent difference} permissible under
Theorem~\ref{theorem:faith-1} while imposing a stricter upper bound on
$e_{b} - e_{a}$ to ensure $x - a \in \mathbb{F}$.

%% file: sec.rto.tex
\section{EFT Conditions for Round-to-Odd}\label{sec:rto}

For any $r \in \mathbb{R}$, Definition~\ref{def:rto} (see Section~\ref{sec:background}) ensures that $x = \ro(r)$ has an even significand only if $r \in \mathbb{F}$. Assuming $p \ge 2$, the definition $\Omega = (2^{p} - 1) \times 2^{E_{max}-p+1}$ implies that the maximum magnitude finite elements in $\mathbb{F}$ have odd significands. For all finite $a, b \in \mathbb{F}$ such that $|a+b| > \Omega$, $a+b$ is not an element of $\mathbb{F}$. \ro would thus round such sums to a neighbor $x \in \mathbb{F}$ for which $|x| = \Omega$. \ro effectively enforces \textit{saturation}, thereby preventing overflow for all sums $x = \ro(a+b)$. Leveraging this property, we present a method to determine $\delta = a+b - \ro(a+b) \in \mathbb{F}$ based on $|a+b|$.

\begin{lemma}\label{lemma:rto-1}
Let $p \ge 2$, $a, b \in \mathbb{F} \setminus \{ \pm \infty \}$, and 
$x = \ro(a+b)$. If $|a+b| > \Omega$, then $\delta = a+b - x \in \mathbb{F}$.
\end{lemma}

\begin{proof}
  If $|a+b| > \Omega$, then $|x| = \Omega$ due to saturation. Without loss of generality, we assume $|a| \ge |b|$. The assumptions $|a+b| > \Omega$ and $|a| \ge |b|$ imply $a \times b > 0$  and $\ufp(a) = 2^{E_{max}}$ as $|a+b| \le \Omega$ otherwise. Since $|a| \ge |b|$ implies $\ufp(a) \ge \ufp(b)$, we derive $a, b \in 2 u \cdot \ufp(b)\mathbb{Z}$. It then follows that $\delta \in 2u \cdot \ufp(b)\mathbb{Z}$.
  
  Since $\ro \in \faith$ and $|x|, |a+b| > 0$, it immediately follows that $x 
  \cdot (a+b) > 0$. Given $|a| \le \Omega$ (\ie, $a$ is finite), 
  $|a+b| > \Omega$, $|x| = \Omega$, and $x \cdot (a+b) > 0$, it then follows 
  that $|\delta| = |a+b - x| \le |a+b| - |x| \le |a| + |b| - \Omega \le |b|$.
  From $|\delta| \le |b|$, we subsequently derive $|\delta| < 2\ufp(b)$. Due to 
  Corollary~\ref{corollary:fp}, $\delta \in 2u \cdot \ufp(b)\mathbb{Z}$ and $|\delta| < 
  2\ufp(b)$ certify $\delta \in \mathbb{F}$.
\end{proof}

Lemma~\ref{lemma:rto-1} signifies that the rounding error of $\circ(a+b)$ is in $\mathbb{F}$ when $|a+b| > \Omega$ and the rounding mode enforces saturation. The bound $|a+b| > \Omega$ is thus also sufficient for $\delta \in \mathbb{F}$ when $x = \rz(a+b)$. We note, however, saturation only guarantees $x$ and $\delta$ are finite; it does not preserve all properties of faithful rounding expected for $|a+b| \le \Omega$. Specifically, saturation does not guarantee $|\delta| < \ulp(a+b)$.

Since $\ro \in \faith$, the conditions presented in Section~\ref{sec:faithful} directly apply to $\ro$. Theorem~\ref{theorem:faith-1} provides the relevant conditions for determining $\delta \in \mathbb{F}$ when $|a+b| \le \Omega$. Combining Theorem~\ref{theorem:faith-1} and Lemma~\ref{lemma:rto-1}, we derive the following conditions for $\ro$ that are independent of the magnitude of $a+b$.
 
\begin{corollary}\label{corollary:rto-1}
Let $p \ge 2$, $a, b \in \mathbb{F} \setminus \{ \pm \infty \}$, and 
$x = \ro(a+b)$. If both conditions

\begin{itemize}
	\item[]{(i)} $a \in 2u^{2} \cdot \ufp(b)\mathbb{Z}$
	\item[]{(ii)} $b \in 2u^{2} \cdot \ufp(a)\mathbb{Z}$
\end{itemize}
are satisfied, then $\delta = a+b - x \in \mathbb{F}$.
\end{corollary}

For directed rounding modes, the rounding direction for a given $r \in \mathbb{R}$ is either fixed (\eg, $\rd$ and $\ru$) or determined by the sign of $r$ (\eg, \rz). \ro, on the other hand, determines the rounding direction for a number based on \textit{the parity of its neighbors' integral significands}. $\ro$'s dependence on the parity of elements in $\mathbb{F}$ induces the following property.

\begin{lemma}\label{lemma:rto-2}
Let $p \ge 2$, $r \in \mathbb{R}$, and $x = M_{x} \times 2^{E_{x}-p+1} 
\in \mathbb{F} \setminus \{ \pm \infty \}$. If both conditions

\begin{itemize}
	\item[]{(i)} $M_{x}$ is an odd integer
	\item[]{(ii)} $\fpred(x) < r <\fsucc(x)$
\end{itemize}
are satisfied, then $x = \ro(r)$.
\end{lemma}

\begin{proof}
	We assume $r \neq x$, as the lemma is trivially true otherwise. We 
	decompose the proof into two cases: $r < x$ and $r > x$. If $r < x$, 
	Condition (ii) ensures $\fpred(x) < r < x$. Because no other numbers in 
	$\mathbb{F}$ exist between $\fpred(x)$ and $x$, $\fpred(x)$ and $x$ are 
	equal to $\fpred(r)$ and $\fsucc(r)$ respectively. Given Condition (i), $x 
	= \ro(r)$ immediately follows from Definition~\ref{def:rto}. Similarly, $r 
	> x$ implies $x = \fpred(r) < r < \fsucc(x) = \fsucc(r)$. Condition (i) and 
	Definition~\ref{def:rto} also induce $x = \ro(r)$ for such cases. 
\end{proof}

We apply Lemma~\ref{lemma:rto-2} to identify the rounding direction for $\ro$ sums when the addends' significands do not overlap (\ie, $\min(|a|, |b|) < \ulp(\max(|a|, |b|))$). Under such circumstances, the larger magnitude operand could be a floating-point neighbor of $a+b$. It would then follow that $r = a+b$ and $|x| = \max(|a|, |b|)$ satisfy Condition (ii). Based on these observations, we present a condition for $\delta \in \mathbb{F}$ that leverages the parity of the operands.

\begin{lemma}\label{lemma:rto-3}
Let $p \ge 2$, $a = M_{a} \times 2^{E_{a}-p+1} \in \mathbb{F} \setminus 
\{ \pm \infty \}$, $b = M_{b} \times 2^{E_{b}-p+1} \in \mathbb{F} \setminus 
\{ \pm \infty \}$, and $x = \ro(a+b)$. If $M_{\max(|a|,|b|)}$ is an odd 
integer, then $\delta = a+b - x \in \mathbb{F}$.
\end{lemma}

\begin{proof}
	If $|a+b| > \Omega$, $\delta \in \mathbb{F}$ holds due to Lemma~\ref{lemma:rto-1}. We thus only consider cases where $|a+b| \le \Omega$. Without loss of generality, we assume $|a| \ge |b|$. Our assumptions thus imply $M_{\max(|a|,|b|)} = M_{|a|}$ and that $M_{a}$ is an odd integer. We proceed with the proof using two cases: $|b| \ge \ulp(a)$ and $|b| < \ulp(a)$.

	\textbf{Case 1: $\bm{|b| \ge \ulp(a)}$.} Since $\ulp(a)$ is an integer power of 2, $b$ satisfies $\ufp(b) \ge \ulp(a)$. Given $b \in u\cdot\ufp(b)\mathbb{Z}$, it follows that $b \in u\cdot\ulp(a)\mathbb{Z}$. Because $\ulp(a) \ge 2u \cdot \ufp(a)$, $b \in u\cdot\ulp(a)\mathbb{Z}$ induces $b \in 2u^{2}\cdot\ufp(a)\mathbb{Z}$. Similarly, $|a| \ge |b|$ guarantees $a \in 2u^{2}\cdot\ufp(b)\mathbb{Z}$. Therefore, $a$ and $b$ satisfy the conditions of Theorem~\ref{theorem:faith-1}, which confirms $\delta \in \mathbb{F}$.

	\textbf{Case 2: $\bm{|b| < \ulp(a)}$.} If $p \ge 2$, all integer powers of 
	2 in $\mathbb{F}$ greater than $\omega$ have even integral significands. 
	The same applies to their negatives. If $M_{a}$ is an odd integer, 
	either $|a| \neq \ufp(a)$ or $|a|= \omega$ must hold. Due to Definition~\ref{def:pred-succ-fp}, the distance between $a$ and its neighbors is thus $|
	\fpred(a) - a| = |\fsucc(a) - a| = \ulp(a)$. The assumption $|b| < \ulp(a)$ 
	would then indicate $\fpred(a) < a+b < \fsucc(a)$. $M_{a}$ being an odd 
	integer and $\fpred(a) < a+b < \fsucc(a)$ match the conditions of Lemma~\ref{lemma:rto-2}, which affirms $x = a$. If $x = a$, then $\delta = a+b - x = b$. 
	It immediately follows that $\delta$ is in $\mathbb{F}$.
\end{proof}

Provided that the operand with the larger magnitude has an odd significand, Lemma~\ref{lemma:rto-3} does not restrict the magnitude of the other operand. As such, the exponent difference of operands that adhere to Lemma~\ref{lemma:rto-3} could span the entire dynamic range of $\mathbb{F}$. The lemma thereby deviates from its counterparts in Section~\ref{sec:background} (see Lemma~\ref{lemma:faith-2} and Theorem~\ref{theorem:faith-1}), which entail bounds on the operands' relative magnitudes. Additionally, Lemma~\ref{lemma:rto-3} does not restrict the sign of the smaller magnitude operand as required for directed rounding modes.

Lemmas~\ref{lemma:rto-1} and~\ref{lemma:rto-3} each leverage $\ro$'s saturation 
and dependence on significand parity to ensure $\delta \in \mathbb{F}$. As a 
result, the lemmas provide guarantees for operands that may not satisfy the 
sufficient conditions for all faithful rounding modes. By applying Lemmas~\ref{lemma:rto-1} and~\ref{lemma:rto-3} in conjunction with exactness conditions 
for FastTwoSum's second operation (\ie, $z = \circ_{2}(x - a)$), we produce the 
following criteria that ensure FastTwoSum is an EFT under $\ro$.

\begin{theorem}\label{theorem:rto-1}
Let $p \ge 2$, $a = M_{a} \times 2^{E_{a}-p+1} \in \mathbb{F} \setminus 
\{ \pm \infty \}$, and $b = M_{b} \times 2^{E_{b}-p+1} \in \mathbb{F} \setminus 
\{ \pm \infty \}$. Let $x = \ro(a+b)$, $z = \circ_{2}(x - a)$, 
and $y = \circ_{3}(b - z)$. If all conditions

\begin{itemize}
	\item[]{(i)} $\circ_{2}, \circ_{3} \in \faith$
	\item[]{(ii)} $a \in \ulp(b)\mathbb{Z}$
	\item[]{(iii)} $M_a$ is an odd integer
\end{itemize}
are satisfied, then $x+y = a+b$.
\end{theorem}

\begin{proof}

Condition (ii) ensures $x - a \in \mathbb{F}$, which then guarantees 
$z = x - a$ for all $\circ_{2} \in \faith$. It therefore suffices to prove 
$\delta = a+b - x \in \mathbb{F}$ as it guarantees $y = a+b - x$ for all 
$\circ_{3} \in \faith$. We thus prove $\delta$ is in $\mathbb{F}$ for all 
operands that meet Conditions (ii) and (iii) by separately analyzing the cases 
$|a| \ge |b|$ and $|a| < |b|$.

	\textbf{Case 1: $\bm{|a| \ge |b|}$.} If $|a| \ge |b|$, Condition (iii) implies $\max(|a|, |b|)$ has an odd significand. The operands thus satisfy the condition in Lemma~\ref{lemma:rto-3}, which affirms $\delta \in \mathbb{F}$.

	\textbf{Case 2: $\bm{|a| < |b|}$.} The assumption $|a| < |b|$ indicates $\ufp(b) \ge \ufp(a)$. From $\ufp(b) \ge \ufp(a)$, we derive $b \in 2u^{2}\cdot\ufp(a)\mathbb{Z}$. Because $\ulp(b) \ge 2u^{2}\cdot\ufp(b)$, Condition (ii) induces $a \in 2u^{2}\cdot\ufp(b)\mathbb{Z}$. The operands that satisfy both $|a| < |b|$ and Condition (ii) thus meet the conditions in Corollary~\ref{corollary:rto-1}, which confirms $\delta \in \mathbb{F}$.
\end{proof}

Theorem~\ref{theorem:rto-1} enables for $\ro$ the use of FastTwoSum on operands 
such that an EFT is not guaranteed under directed rounding modes. For example, 
suppose $a = 2^{p} + 2$ and $b = -2^{-p}$. In this example, $a$ is an integer 
multiple of $\ulp(b) = 2^{1-2p}$ and has an odd significand (\ie, $a = 2^{p} + 
2 = (2^{p-1} + 1) \times 2$). The operands therefore satisfy the conditions in 
Theorem~\ref{theorem:rto-1}. However, $b$ is not an integer multiple of $2u^{2} 
\cdot \ufp(a) = 2^{1-p}$ as required by Theorem~\ref{theorem:faith-2}. If 
$\circ_{1} = \circ_{2} = \circ_{3} = \rz$, $x = \rz(a+b) = 2^{p}$ and $z = \rz
(x - a) = -2$. As a result, $y = \rz(b - z) = 2 - 2^{1-p}$, and $x + y \neq a + 
b$. Under $\ro$, FastTwoSum produces $x = 2^{p} + 2$, $z = 0$, and 
$y = -2^{-p}$. Hence, FastTwoSum under $\ro$ achieves $x+y = a+b$ as guaranteed 
by Theorem~\ref{theorem:rto-1}.

%% file: sec.applications.tex
\section{Applications of Round to Odd FastTwoSum}\label{sec:applications}

A common application of FastTwoSum is floating-point splitting: an EFT that splits $x \in \mathbb{F}$ across two numbers $x_{h}, x_{\ell} \in \mathbb{F}$ such that $x = x_{h} + x_{\ell}$. Dekker's splitting~\cite{dekker:fts:1971}, for example, splits a $p$-precision number across two numbers that each have up to $\lfloor \frac{p}{2} \rfloor$ effective bits of precision. In~\cite{rump:accsum-1:siam:2008}, the authors present a more general splitting known as ExtractScalar, for which the distribution of bits is configurable.

\begin{empheq}[box=\fbox]{align*}
\textbf{Ex}&\textbf{tractScalar}(\sigma, x):\\
&s = \circ_{1}(\sigma + x)\\
&x_{h} = \circ_{2}(s - \sigma)\\
&x_{\ell} = \circ_{3}(x - x_{h})\\
&\textbf{return } x_{h}, x_{\ell}
\end{empheq}

ExtractScalar applies FastTwoSum on $x, \sigma \in \mathbb{F}$ such that 
$\sigma = 2^{k}$ for some $k \in \mathbb{Z}$ (\ie, $\sigma$ is an integer power 
of 2). It is shown in~\cite{rump:accsum-1:siam:2008} that if $|x| \le 2^{k}$ 
and $\circ_{1}, \circ_{2}, \circ_{3} \in \rn$, then $x_{h} \in \frac{1}{2}\ulp
(\sigma)\mathbb{Z}$ and $x = x_{h} + x_{\ell}$. If $x_{h} \in \frac{1}{2}\ulp
(\sigma)\mathbb{Z}$, the number of effective bits in $x_{h}$ is determined by 
the exponent difference between $\sigma$ and $x$. Hence, ExtractScalar can 
adjustably distribute bits across $x_{h}$ and $x_{\ell}$ by assigning an 
appropriate value of $\sigma$ relative to $x$. Given a vector $v \in 
\mathbb{F}^{n}$, ExtractScalar can also produce an EFT for \textit{all 
elements} $v_{i}$ if $\sigma$ is configured relative to $\max_{i}|v_{i}|$. 

In~\cite{rump:accsum-1:siam:2008}, ExtractScalar performs all operations under $\rn$ while assuming overflow does not occur. Furthermore, the assumptions $\sigma = 2^{k}$ and $|x| \le 2^{k}$ do not immediately satisfy the EFT conditions for FastTwoSum detailed in Section~\ref{sec:faithful}. As such, ExtractScalar may not be an EFT when $\circ_{1} \not \in \rn$. Specifically, an EFT is not guaranteed when the exponent difference between $\sigma$ and $x$ exceeds $2p - 1$. Suppose $\circ_{1}, \circ_{2}, \circ_{3} = \ro$ and $x = 2^{k-2p}$. In this example, $s = \ro(\sigma + x) = 2^{k} + 2^{k-p+1}$ and $x_{h} = \ro(2^{k} + 2^{k-p+1} - 2^{k}) = 2^{k-p+1}$. Subsequently, $x_{\ell} = \ro(2^{k-2p} - 2^{k-p+1}) = 2^{k-2p+1} -2^{k-p+1}$ and $x \neq x_{h} + x_{\ell}$. Similar examples can be constructed for directed rounding modes 
(\eg, $x = 2^{k-2p}$ for $\ru$ and $x = -2^{k-2p-1}$ for $\rd$ and $\rz$). Given these limitations, we present new conditions that ensure ExtractScalar is an EFT under $\ro$.

\begin{theorem}\label{theorem:extract-scalar}
Let $p \ge 2$, $x, \sigma \in \mathbb{F} \setminus \{\pm \infty\}$, and $k \in 
\mathbb{Z}$. Let $s = \ro(\sigma+x)$, $x_{h} = \circ_{2}(s - \sigma)$, and 
$x_{\ell} = \circ_{3}(x - x_{h})$. If all conditions

\begin{itemize}
	\item[]{(i)} $\circ_{2}, \circ_{3} \in \faith$
    \item[]{(ii)} $\sigma = 2^{k} + \ulp(2^{k})$ and $2^{k} \ge 2\omega$ 
	\item[]{(iii)} $|x| \le 2^{k}$
\end{itemize}
are satisfied, then $x_{h} \in \frac{1}{2}\ulp(\sigma)\mathbb{Z}$ and $x = x_{h} + x_{\ell}$.
\end{theorem}

\begin{proof}
    Given $p \ge 2$ and Condition (ii), we derive $\ufp(\sigma) = 2^{k}$ and 
    $\ulp(\sigma) = \ulp(2^{k})$. Through Conditions (ii) and (iii), we also 
    derive $|x| \le \ufp(\sigma) < \sigma$. We first prove $x_{h} \in 
    \frac{1}{2}\ulp(\sigma)\mathbb{Z}$. If $s, \sigma \in \frac{1}{2}\ulp
    (\sigma)\mathbb{Z}$, it must be the case that $x_{h} = \circ_{2}(s - 
    \sigma) \in \frac{1}{2}\ulp(\sigma)\mathbb{Z}$. As $\sigma \in \frac{1}{2}
    \ulp(\sigma)\mathbb{Z}$ by default, it suffices to show $s \in \frac{1}{2}
    \ulp(\sigma)\mathbb{Z}$. We split our proof into three cases: $x \ge 0$, 
    $-2^{k-1} \le x < 0$, and $x < -2^{k-1}$. 
    
    We further decompose $x \ge 0$ into two cases: $\sigma + x > \Omega$ and $\sigma + x \le \Omega$. If $\sigma + x > \Omega$, then $s = \ro(\sigma + x) = \Omega = (2^{p} - 1) \times 2^{E_{max} -p + 1} $ due to $\ro$'s saturation property. Because $|x| \le \ufp(\sigma)$, $\sigma + x > \Omega$ implies $\ufp(\sigma) = 2^{E_{max}}$ and $\ulp(\sigma) = 2^{E_{max}-p+1}$. It then follows that $\ulp(s) = \ulp(\sigma)$. Since $s \in \ulp(s)\mathbb{Z}$ and $\ulp(s) = \ulp(\sigma) > \frac{1}{2}\ulp(\sigma)$, we conclude $s \in \frac{1}{2}\ulp(\sigma)\mathbb{Z}$ when $\sigma + x > \Omega$. If $\sigma + x \le \Omega$, $\ufp(s) = \ufp(\sigma + x)$ (\ie, $e_{s} = e_{\sigma+x}$) for $s = \ro(\sigma + x)$. Furthermore, $x \ge 0$ indicates $\sigma + x \ge \sigma$ and $\ufp(\sigma + x) \ge \ufp(\sigma)$. Given $\ufp(s) = \ufp(\sigma + x) \ge \ufp(\sigma)$, Definition~\ref{def:ufp-ulp} ensures $\ulp(s) = \ulp(\sigma + x) \ge \ulp(\sigma) > \frac{1}{2}\ulp(\sigma)$. As $s \in \ulp(s)\mathbb{Z}$, $s$ thus satisfies $s \in \frac{1}{2}\ulp(\sigma)\mathbb{Z}$ when $x + \sigma \le \Omega$. 
    
    If $-2^{k-1} \le x < 0$ and $s = \ro(\sigma + x)$, then $\sigma + x \ge 2^{k-1} + \ulp(2^{k})$ and $\ufp(s) = \ufp(\sigma + x) \ge 2^{k-1} = \frac{1}{2}\ufp(\sigma)$. As such, $\ulp(s) \ge \frac{1}{2}\ulp(\sigma)$ and $s \in \frac{1}{2}\ulp(\sigma)\mathbb{Z}$. Lastly, $x < -2^{k-1}$ implies $\ufp(x) \ge 2^{k-1} = \frac{1}{2}\ufp(\sigma)$ and $\ulp(x) \ge \frac{1}{2}\ulp(\sigma)$. Since $s = \ro(\sigma + x)$ and $\sigma + x, s \in \min(\ulp(\sigma), \ulp(x))\mathbb{Z}$, it follows that $s \in \frac{1}{2}\ulp(\sigma)\mathbb{Z}$. Hence, $x_{h} \in \frac{1}{2}\ulp(\sigma)\mathbb{Z}$ for all cases considered.
    
    Due to Definition~\ref{def:pred-succ-fp}, $\sigma = 2^{k} + \ulp(2^{k})$ in Condition (ii) is equal to $\fsucc(2^{k})$. All integer powers of 2 in $\mathbb{F}$ greater than $\omega$ have even significands. Consequently, any $\sigma$ that satisfies Condition (ii) must have an odd significand. Because $|x| < \sigma$, $\sigma \in \ulp(x)\mathbb{Z}$. All $\sigma, x \in \mathbb{F}$ that satisfy these conditions thus meet the requirements of Theorem~\ref{theorem:rto-1}. Through Theorem~\ref{theorem:rto-1}, we derive $\sigma + x = s + x_{\ell}$. Because $\sigma \in \ulp(x)\mathbb{Z}$ ensures $x_{h} = s - \sigma$ (see Section~\ref{sec:background}), $\sigma + x = s + x_{\ell}$ induces $x = s - \sigma + x_{\ell} = x_{h} + x_{\ell}$ as was to be shown.
\end{proof}

If the requirements of Theorem~\ref{theorem:extract-scalar} are met, 
ExtractScalar ensures error-free splitting under $\ro$ even when $\sigma$ and 
$x$ are non-overlapping (\ie, $\ulp(\sigma) > |x|$). Given a $\sigma$ that 
satisfies the conditions in Theorem~\ref{theorem:extract-scalar} for 
$x = \max_{i}|v_{i}|$ in a vector $v \in \mathbb{F}^{n}$, ExtractScalar can 
thus produce an EFT for all elements $v_{i}$ under $\ro$ without requiring a 
lower bound on $\min_{i}|v_{i}|$.

%% file: sec.related.tex
\section{Related Work}
\label{sec:related}

In this section, we discuss prior work in the literature that have
also explored EFT guarantees for FastTwoSum under various faithful 
rounding modes. Jeannerod and Zimmermann propose $a \in \ulp(b)\mathbb{Z}$ and 
$e_{a} - e_{b} \le p$ as sufficient conditions for ensuring FastTwoSum is an 
EFT under all faithful rounding modes~\cite[Theorem 2]
{jeannerod:fts-revisited:arith:2025}. Operands that satisfy 
$e_{a} - e_{b} \le p$ must adhere to the condition $b \in
2u^{2} \cdot \ufp(a)\mathbb{Z}$. Accordingly, Theorem~\ref{theorem:faith-2} 
of this paper ensures $x+y = a + b$ for any operands that meet the conditions
from~\cite[Theorem 2]{jeannerod:fts-revisited:arith:2025}. 
Theorem~\ref{theorem:faith-2} of this paper also guarantees EFTs for operands 
that would be excluded by the bound $e_{a} - e_{b} \le p$ (\eg, $a = 2^{p}$ and 
$b = 2^{-1}$). As such, our conditions are applicable to a larger set of 
operands than those addressed by previous works.

In~\cite[Theorems 9, 10, 11]{jeannerod:fts-revisited:arith:2025}, Jeannerod and 
Zimmermann establish that when $a \in \ulp(b)\mathbb{Z}$ but $e_{a} - e_{b} 
> p$, FastTwoSum under directed rounding modes produces a faithful rounding of 
$a+b$ for a format with twice the current available precision 
(\eg, $x+y = \rz_{2p}(a+b)$). These theorems signify that FastTwoSum under 
directed rounding modes is an EFT if $a \in \ulp(b)\mathbb{Z}$ and the real sum 
is exactly representable in $2p$-bits of precision. Our conditions in 
Theorem~\ref{theorem:faith-2} ensure that $a+b$ is exactly representable in 
$2p$-bits and thus guarantee EFTs for directed rounding modes under the same 
circumstances as~\cite[Theorems 9, 10, 11]{jeannerod:fts-revisited:arith:2025}.

For operands that adhere to Theorem~\ref{theorem:faith-1}, the
exponent difference could be as large as $2p - 1$. Conditions with
similar implications appear in~\cite{linnainmaa:fp-analysis:1974}. 
In~\cite{linnainmaa:fp-analysis:1974}, Linnainmaa presents conditions for 
$\delta \in \mathbb{F}$ under \rz and \ro using the notation $h_{y}$: the 
exponent of the \textit{least significant nonzero bit} of $y \in \mathbb{F} \setminus\{0\}$. 
Alternatively, $h_{y}$ is the \textit{largest} integer $k$
such that $y \in 2^{k}\mathbb{Z}$. Given this definition, all finite,
nonzero $y \in \mathbb{F}$ satisfy $y \in 2^{h_{y}}\mathbb{Z}$ and
$\ulp(y) \le 2^{h_y} \le \ufp(y)$. Moreover, if $|y| = \ufp(y)$ (\ie,
$|y|$ is an integer power of 2), then $h_{y} = e_{y}$. Linnainmaa
proves that if $x = \circ(a+b)$ for $\circ \in \{\rz, \ro\}$, $e_{x} -
h_{\min(|a|, |b|)} < 2p$ guarantees $\delta = a+b - x \in
\mathbb{F}$. If $\ulp(a) > |b|$ (\ie, $|a| > |b|$ and the operands' 
significands do not overlap) and $a \times b > 0$, $e_{x} = e_{a}$ must hold 
under $\circ \in \{\rz, \ro \}$ due to Definitions~\ref{def:ufp-ulp} and
~\ref{def:pred-succ-fp}. The condition $e_{x} - h_{\min(|a|,
  |b|)} < 2p$ for such cases is thus equivalent to $e_{a} - h_{|b|} <
2p$. Since $h_{|b|} = h_{b} = e_{b}$ when $|b| = \ufp(b)$, the example
indicates that Linnainmaa's condition can tolerate up to a $2p - 1$
difference in exponents.

While Linnainmaa's condition provides comparable flexibility under
$\rz$ and $\ro$, it lacks precision under $\rd$ and $\ru$. For
example, consider operands $a = 2^{p} - 1$ and $b = 2^{-p}$. Given
$|a| \ge |b|$ and $2u^{2} \cdot \ufp(a) = 2^{-p}$, the operands
satisfy $a \in 2u^{2}\cdot \ufp(b)\mathbb{Z}$ and $b \in 2u^{2}\cdot
\ufp(a)\mathbb{Z}$ as required by Theorem~\ref{theorem:faith-1}. If
$\circ = \ru$, $x = \circ(a+b) = 2^{p}$ and $\delta = 2^{-p} - 1 \in
\mathbb{F}$. In this example, $e_{x} = p$, $h_{\min(|a|,|b|)} =
-p$, and $e_{x} - h_{\min(|a|, |b|)} = 2p$. As such, the example
operands do not satisfy Linnainmaa's condition despite $\delta$ being
representable in $\mathbb{F}$. One can construct a similar example for
$\rd$ by reversing the signs of $a$ and $b$. Hence,
Theorem~\ref{theorem:faith-1} is applicable to a broader range of
operands under $\rd$ and $\ru$ while addressing all faithful rounding
modes including $\rz$ and $\ro$.

Akin to Lemma~\ref{lemma:rto-3}, Linnainmaa proposes
$M_{\max(|a|, |b|)}$ having an odd significand as a sufficient
condition for $\delta \in \mathbb{F}$ under $\ro$. However,
Linnainmaa's analysis assumes an unbounded model of $\mathbb{F}$ and
does not address guarantees due to $\ro$'s saturation property (see
Lemma~\ref{lemma:rto-1}). Furthermore, Linnainmaa proposes $|a| \ge
|b|$ and $a$ having an odd significand as EFT guarantees for
FastTwoSum under $\ro$. Instead of $|a| \ge |b|$,
Theorem~\ref{theorem:rto-1} enforces the less restrictive $a \in
\ulp(b)\mathbb{Z}$ to ensure the second operation of FastTwoSum is
exact.

%% file: sec.conclusion.tex
\section{Conclusion and Future Work}
\label{sec:conclusion}

This paper identifies more general conditions than previously known in
literature that ensure FastTwoSum is an EFT for all faithful rounding
modes. Specifically, we identify new properties of operands that
ensure the rounding error $\delta = a+b - \circ(a+b)$ is in
$\mathbb{F}$. Our conditions enable EFTs for operands with large magnitude
differences for all faithful rounding modes (even when the operands'
exponent difference is nearly double the available precision). Hence,
our EFT guarantees for FastTwoSum are applicable to a
wide range of inputs, thereby offering improved precision over
previously established conditions.

This paper also presents EFT conditions tailored to $\ro$. We leverage
$\ro$'s parity-based rounding behavior to identify sufficient
requirements for both $\delta \in \mathbb{F}$ and $x+y = a+b$. We
observe that when the larger operand has an odd significand, 
FastTwoSum can serve as an EFT under $\ro$ without restricting 
the magnitude or the sign of the smaller operand, which
highlights $\ro$'s versatility over directed rounding modes. By
analyzing a bounded floating-point model, we also examine $\ro$'s
saturation property and its enabling of EFTs. We leverage these
findings to identify conditions for error-free floating-point
splittings under $\ro$.

FastTwoSum-based splitting facilitates more sophisticated EFTs such as
integer rounding~\cite{jeannerod:fp-split:arith:2018}, correctly
rounded
summation~\cite{rump:accsum-1:siam:2008}~\cite{rump:accsum-2:siam:2009},
and accurate dot-products~\cite{ozaki:ozaki-scheme:2012}. To leverage
EFT guarantees, most algorithms based on FastTwoSum require $\rn$ as
the default rounding mode. We expect our findings to aid the
development of advanced EFTs for other standard rounding modes as well
as emerging alternatives such as $\ro$. Applying our conditions to
EFTs for floating-point multiplication is also of interest. In a
similar vein, extending our observations on rounding errors induced by
floating-point addition to fused multiply-add operations is a
potential future direction.